\documentclass[jphysa,amsmath]{iopart}

\usepackage{amsmathred}
\usepackage{graphicx}
\usepackage{mathrsfs}
\usepackage{dsfont}
\usepackage{amssymb}
\usepackage{color}
\usepackage{blkarray}
\def\bra#1{{\left\langle #1 \right|}}
\def\ket#1{{\left| #1 \right\rangle}}
\def\braket#1{{\left\langle #1 \right\rangle}}
\definecolor{amber}{rgb}{1.0, 0.49, 0.0}
\usepackage[dvipsnames]{xcolor}

\usepackage{verbatim}
\usepackage{amsthm}
\usepackage{enumerate}
\usepackage{bbm}
\usepackage{braket}
\usepackage{scalerel}
\usepackage{stackengine}
\usepackage{times,txfonts}
\usepackage[breaklinks=true,colorlinks=true,linkcolor=blue,urlcolor=blue,citecolor=blue]{hyperref}
\usepackage{cite}

\newtheorem{theorem}{Theorem}
\newtheorem{tmain}[theorem]{Theorem}
\newtheorem{example}[theorem]{Example}
\newtheorem{lemma}[theorem]{Lemma}
\newtheorem{corollary}[theorem]{Corollary}
\newtheorem{remark}[theorem]{Remark}

\newtheorem{proposition}[theorem]{Proposition}
\newtheorem{definition}[theorem]{Definition}

\newcommand*{\sumno}{\sum\nolimits}

\newcommand*{\bbE}{\mathbb{E}}
\newcommand*{\bbI}{\mathbb{I}}

\newcommand*{\cE}{\mathcal{E}}

\newcommand*{\cM}{\mathcal{M}}

\newcommand*{\ketbra}[1]{\ket{#1}\!\!\bra{#1}}

\newcommand*{\proj}[1]{\ket{#1}\bra{#1}}
\renewcommand*{\Tr}{\mathrm{Tr}}

\renewcommand*{\coloneqq}{\mathrel{\vcenter{\baselineskip0.5ex \lineskiplimit0pt \hbox{\scriptsize.}\hbox{\scriptsize.}}} =}

\newcommand{\beq}{\begin{equation}}
\newcommand{\eeq}{\end{equation}}
\newcommand{\beqq}{\begin{equation*}}
\newcommand{\eeqq}{\end{equation*}}

\DeclareMathOperator{\id}{id}

\makeatletter
\def\thmhead@plain#1#2#3{%
  \thmname{#1}\thmnumber{\@ifnotempty{#1}{ }\@upn{#2}}%
  \thmnote{{ \the\thm@notefont#3}}}
\let\thmhead\thmhead@plain
\makeatother

\newtheorem{manualtmaininner}{Theorem}
\newenvironment{manualtmain}[1]{%
  \renewcommand\themanualtmaininner{#1}%
  \manualtmaininner
}{\endmanualtmaininner}

\newtheorem{manualdefinitioninner}{Definition}
\newenvironment{manualdefinition}[1]{%
  \renewcommand\themanualdefinitioninner{#1}%
  \manualdefinitioninner
}{\endmanualdefinitioninner}

\begin{document}

\title{Refined diamond norm bounds on the emergence of objectivity of observables 
}


\author{Eugenia Colafranceschi${}^1$, Ludovico Lami${}^{1,2}$, Gerardo Adesso${}^1$, and Tommaso Tufarelli${}^1$}
\address{${}^1$School of Mathematical Sciences and Centre for the Mathematics and Theoretical Physics of Quantum Non-Equilibrium Systems, University of Nottingham, University Park Campus, Nottingham NG7 2RD, United Kingdom}
\address{${}^2$Institut f\"{u}r Theoretische Physik und IQST, Albert-Einstein-Allee 11, Universit\"{a}t Ulm, D-89081 Ulm, Germany}
\ead{eugenia.colafranceschi@nottingham.ac.uk}

\begin{abstract}
    The theory of Quantum Darwinism aims to explain how our objective classical reality arises from the quantum world, by analysing the distribution of information about a quantum system that is accessible to multiple observers, who probe the system by intercepting fragments of its environment. Previous work 
    showed that, when the number of environmental fragments grows, the quantum channels modelling the information flow from system to observers become arbitrarily close -- in terms of diamond norm distance -- to ``measure-and-prepare'' channels, ensuring objectivity of observables; the convergence is formalised by an upper bound on the diamond norm distance, which decreases with increasing number of fragments. Here, we derive tighter diamond norm bounds on the emergence of objectivity of observables for quantum systems of infinite dimension, providing an approach which can bridge between the finite- and the infinite-dimensional cases. Furthermore, we probe the tightness of our bounds by considering a specific model of a system-environment dynamics given by a pure loss channel. Finally, we generalise to infinite dimensions a result obtained by Brand\~{a}o {\it et al.} [{\it Nat.~Commun.}~{\bf 6}, 7908 (2015)], 
    which provides an operational characterisation of quantum discord in terms of one-sided redistribution of correlations to many parties. Our results provide a unifying framework to benchmark quantitatively the rise of objectivity in the quantum-to-classical transition.
\end{abstract}


\section{Introduction}

Quantum theory has proven to be extremely successful in describing the physical laws of microscopic objects. However, assuming the general validity of quantum theory, the apparent absence of quantum features (such as non-locality and superposition effects) in our everyday classical reality raises the issue of the quantum-to-classical transition: how do physical systems lose their ``quantumness'' with increasing scales and become effectively classical?

 The theory of decoherence~\cite{joos2003decoherence,RevModPhys.75.715,RevModPhys.76.1267,schlosshauer2007decoherence}, which developed significantly over the past decades, has pointed out the key role played in this transition by the interaction of the system with its environment: due to this interaction the two can become entangled, and the quantum correlations so established between the two parties cannot be observed at the level of the system alone. 
 The entanglement with the environment thus \textit{defines} the physical properties we can observe at the level of the system. In particular, only those states that are robust in spite of the interaction with the environment are observable in practice. The environmental monitoring therefore leads to the selection of preferred states (known as \textit{pointer states}~\cite{PhysRevD.24.1516,PhysRevD.26.1862,pointer}) which represent the natural candidates for the classical states that are compatible with our everyday experience. However, decoherence alone does not explain how the striking contrast between classical and quantum states is overcome in the emergence of classicality.
In fact, while classical states can be detected and agreed upon by initially ignorant observers without being perturbed, and thus exist objectively, quantum states are generally affected by the measurement process. It is therefore necessary to clarify how the information about pointer states becomes objective.

 The theory of Quantum Darwinism~\cite{zurek2009quantum,PhysRevLett.93.220401,PhysRevA.72.042113,PhysRevA.73.062310,PhysRevA.91.032122,TK} provides a possible solution by promoting the environment from source of decoherence to carrier of information about the system. In fact, Quantum Darwinism points out that a fundamental consequence of the system-environment interaction is the presence of information about the system encoded in the environment. By intercepting fragments of the environment, it is possible to acquire such information indirectly. In particular, Quantum Darwinism explains how information about the pointer states proliferates in the environment, allowing multiple observers to detect these states without perturbing their existence.

 The Quantum Darwinism approach to the emergence of classicality has been explored theoretically in various specific models~\cite{Blume2008,Zwolak2009,Riedel2010,Riedel2012,Galve2015,Balaneskovic2015,Tuziemski2015,Tuziemski2015b,Tuziemski2016,Balaneskovic2016,Lampo2017,Pleasance2017,le2018objectivity} and has also been the subject of recent experimental tests~\cite{chen2019emergence,ciampini2018experimental,unden2019revealing}. However, the range of applicability of such framework still represents an open issue in the quantum-to-classical transition problem. A recent result by Brand\~{a}o {\it et al.}~\cite{Brandao2015QD} made a significant contribution to it, showing how some classical features emerge in a model-independent way from the quantum formalism alone. Such result relies on the splitting of the objectivity notion into the following statements:
 \begin{itemize}
     \item[-]{\em Objectivity of observables:} multiple observers probing the same system can at most acquire classical information about one and the same measurement;
\item[-]{\em Objectivity of outcomes:} the observers will agree on the result obtained from the preferred measurement.
\end{itemize}

{Brand\~{a}o {\it et al.}\ modeled the information flow from a (finite-dimensional) quantum system to the fragments of its environment via quantum channels, i.e., completely positive trace-preserving (cptp) maps. They showed that, when the number of fragments $N$ becomes large enough, most of these channels are well approximated by specific cptp maps, called ``measure-and-prepare''. The form of such channels ensures the objectivity requirements. This is formalised by a bound on the distance (induced by the so-called diamond norm) between the system-environment channels and the measure-and-prepare ones. It is found that such distance goes to zero as $N \rightarrow \infty$, leading to convergence to objectivity of observables (in the following, we will refer to such a bound as \textit{objectivity bound}).
In~\cite{emergenza}, Knott {\it et al.}\ overcame the finite-dimension restriction by showing that also infinite-dimensional systems, under appropriate energy constraints, exhibit objectivity of observables. Another interesting result in this context was recently obtained by Qi and Ranard~\cite{qi2020emergent}: they showed that, for finite-dimensional systems, the set of channels which do not converge to objectivity is of fixed size $O(1)$, instead of scaling with the number of environmental fragments $N$, as in ~\cite{Brandao2015QD, emergenza}. Their result, which incorporates insights from an earlier version of  the present manuscript, provides to the best of our knowledge the tightest objectivity bound in the finite-dimensional scenario.
 
 In this paper we extend the infinite-dimensional analysis of~\cite{emergenza} and  provide a unified approach to the emergence of the {\em objectivity of observables} in the interaction between a quantum system of {\em arbitrary} dimension and a large number of fragments of its environment.
 
Specifically, we first prove that the objectivity of observables holds true for a wide class of (energy-constrained) infinite-dimensional systems. For such class we obtain tighter bounds on the emergence of this classical feature, compared to those available in the literature. Moreover, our framework can act as a bridge between the finite- and infinite-dimensional scenarios. Our results rely on an infinite-dimensional version of the Choi--Jamio\l kowski isomorphism, adapted to our set of energy-constrained states. This generalises what was done in~\cite{emergenza} for a specific choice of the energy constraint. Moreover, our analysis exploits novel bounds relating the diamond-norm distance of two channels with the distance between their respective Choi--Jamio\l kowski states --~see Aubrun {\it et al.}~\cite{XOR}. Such results are presented in Section~\ref{newbound}.

A relevant issue concerning the emergence of objectivity of observable, not tackled in Refs~\cite{Brandao2015QD, emergenza}, concerns the optimality of the rates at which the convergence to objectivity takes place. In fact, objectivity of observables is regarded as emergent whenever the upper bound on the distance between channels representing the system-environment information flow and the measure-and-prepare ones goes to zero asymptotically. But this does not give information on how well the objectivity bound approximates the considered diamond norm distance. To perform such optimality check, a possible strategy is to derive a \textit{lower bound} for that diamond norm, which turns out to be an upper bound on the speed at which the emergence of objectivity of observables takes place. In Section~\ref{attenuator} we perform this analysis for the specific model of a system-environment dynamics given by a pure loss channel, and show that for such model the rate of convergence to objectivity of observables scales at least as the inverse of the number of environmental fragments.

The final point we address regards the extension to an infinite-dimensional scenario of the operational interpretation of {\it quantum discord}~\cite{ollivier2001quantum,henderson2001classical} derived for finite-dimensional systems by Brand\~{a}o {\it et al.}~\cite{Brandao2015QD}. In particular, it was proven that when information is distributed to many parties on one side of a bipartite system, the minimal average loss in correlations corresponds to the quantum discord. In Section~\ref{discord} we generalise this result to the infinite-dimensional case by exploiting the objectivity bounds proved in Section~\ref{newbound}.

In summary, the paper is organised as follows. Our improved objectivity bounds are presented in Section~\ref{newbound}, followed by the pure loss channel analysis in Section~\ref{attenuator}, while the operational interpretation of quantum discord is found in Section~\ref{discord}. Some technical details behind our proofs are deferred to the Appendixes.

\section{Improved bounds on the emergence of objectivity of observables}
\label{newbound}
The scenario we consider consists of a system $A$, generally infinite-dimensional, and its environment $B$, which is described as a collection of $N$ (possibly infinite-dimensional) subsystems $B_1,...,B_N$, namely the environment fragments. We assume that the system of interest $A$ is initially decorrelated from $B_1,...,B_N$, and that the corresponding state has bounded mean energy (defined via an appropriate Hamiltonian -- see below). The information flow from the system to the whole environment is modelled as a quantum channel, i.e., a cptp map $\Lambda:\mathcal{D}(A) \rightarrow \mathcal{D}(B_1 \otimes \ldots \otimes B_N)$, where $\mathcal{D}(Z)$ denotes the set of density matrices associated with a physical system $Z$. The transfer of quantum information from $A$ to the single environmental fragment $B_j$ is therefore described by the ``subchannel"
$\Lambda_j = \Tr_{ B\backslash B_j }\circ\Lambda$. Objectivity of observables then arises whenever the  maps $\Lambda_j$ become arbitrarily close to measure-and-prepare channels, which allow observers to acquire only classical information about one and the same measurement. These channels are defined as $\cE_j(X) \coloneqq  \sum_l   \Tr (M_l X) \tau_{j,l}$, where $\{M_l\}_l$ is a positive operator-valued measure (POVM) -- crucially independent of the index $j$ -- and $\{\tau_{j,l}\}_l$ is a set of states for subsystem $B_j$.

We shall quantify distinguishability in the space of channels via a distance called {\em energy-constrained diamond norm}~\cite{Shirokov2018,VV-diamond}. This is a modification of the standard diamond norm~\cite{Aharonov1998,Sacchi2005,Watrous2009}, designed to implement a restriction on the average (initial) energy of the quantum system under examination. This is measured by a Hamiltonian, which we take to be an arbitrary self-adjoint operator $H$ with spectrum bounded from below. Without loss of generality, we assume its ground state energy to be positive, i.e.
\begin{equation}
    \inf_{\lambda\in \mathrm{sp}(H)} \lambda = E_0 >0\, ,
    \label{assumption}
\end{equation}
where $\mathrm{sp}(H)$ is the spectrum of $H$. 

\begin{definition} \label{diamondH}
Let $A'$ be a quantum system equipped with a Hamiltonian $H_{A'}$ that satisfies Eq.~\eqref{assumption}, and pick $E > E_0$. Then the energy-constrained diamond norm of an arbitrary Hermiticity-preserving linear map $\Lambda : \mathcal{D}(A') \rightarrow \mathcal{D}(B)$ is defined by
	\beq 
	\|\Lambda\|_{\Diamond H,E}\coloneqq
	\sup_{
		\substack{
			\Tr\left[ \rho H_{A'} \right] \leq E
		}
	} \| (\id_A \otimes \Lambda_{A'})[\rho_{AA'}]\|_1\ ,
	\label{ec diamond norm}
	\eeq
where $A$ is an arbitrary ancillary system, and $ \| \cdot \|_1$ is the one-norm. A recent result by Weis and Shirokov~\cite{Weis-Shirokov} ensures that the input state $\rho_{AA'}$ in Eq.~\eqref{ec diamond norm} can be taken to be pure.
\end{definition}
\noindent

In our analysis, we assume that the Hamiltonian admits a countable set of eigenvectors forming an orthonormal basis $\{\ket{j}, j=0,1, \dots\}$ of the Hilbert space. The index $j$ is allowed to go to infinity, and the case of a finite-dimensional system will be treated by a suitable choice of the Hamiltonian eigenvalues (see below). \\We want to stress that the assumption that a Hamiltonian $H$ has discrete spectrum and is bounded from below is physically well motivated; in fact, it is contained in the so-called Gibbs hypothesis~\cite{tightuniform}:

\medskip

\textit{Gibbs hypothesis.}
A (possibly unbounded) self-adjoint operator $H$ is said to satisfy the Gibbs hypothesis if for every $\beta > 0$ the partition function $Z(\beta) \coloneqq \Tr e^{-\beta H}$ is finite. As a consequence, the state ${1 \over Z(\beta)} e^{-\beta H}$ has finite entropy. Moreover, for every eigenvalue $E$ of the Hamiltonian $H$, the (unique) maximiser $\rho$ of the entropy subjected to the constraint $\Tr \rho H \le E$ is the Gibbs state
 \beq \label{gammamain}
 \gamma(E)= {1 \over Z(\beta (E))} e^{-\beta(E) H},
 \eeq
where $\beta=\beta(E)$ is the solution to the equation $\Tr e^{-\beta H}(H -E) = 0$. 

By setting \begin{equation}
H=\sum_j f_j \ketbra{j}\, ,
\label{H-assumption}
\end{equation}
we also require that the (increasing) sequence of eigenvalues $f_j$ diverges sufficiently rapidly, in formula
\begin{equation}\label{log-condizione}
\left|\sum_j \frac{1}{f_j}\log \frac{1}{f_j}\right|<\infty .
\end{equation}
In the following, we will refer to the entire spectrum $\{f_j\}_j$ by using the short notation $f= \{f_j\}_j$. Eq.~\eqref{log-condizione} clearly implies that $\sum_j \frac{1}{f_j}<\infty$. Notably,  this excludes the physically relevant case $f_j=j$, corresponding to the canonical Hamiltonian on the Hilbert space of a harmonic oscillator. In spite of this drawback, our technical assumption allows us to explore a rich family of constraints that effectively extend and interpolate between previously known bounds. Moreover, a slight modification of our proof technique allows us to deal with the excluded case $f_j=j$ as well; for details, see the end of this section.
\\Before proceeding with the presentation of our results, we want to clarify that the positive operator $H$ and the scalar threshold $E$ do not actually have any dynamical characterization. In fact, their role in our framework is simply to confine the bulk of a state's probability mass to a finite-dimensional subspace in a smooth and well-defined way.

We now introduce some technical elements and definitions that will enter our main results on the emergence of objectivity of observables, stated in Theorem~\ref{main_th_f}. We start by considering a special class of entangled states featuring an $f$-dependent tail in the Hamiltonian eigenbasis:
 \beq \label{phi}
 \ket{\phi} \coloneqq c_f \sum_{j=0}^{\infty} \phi_j \ket{j,j}_{AA'}\, ,
 \eeq
with $\phi_j^2 \coloneqq 1 / f_j$ and $c_f \coloneqq \left(\sum \frac{1}{f_j}\right)^{-\frac12}$. In our derivation, an important role will be played by the local von Neumann entropy of $\ket{\phi}$, given by
\beq\label{entropy}
\sigma \coloneqq S\big(\Tr_{A'}\ketbra{\phi}_{AA'}\big) = - \sum_j\frac{c_f^2}{f_j} \log \left(\frac{c_f^2}{f_j}\right) < \infty ,
\eeq
where the last inequality follows from Eq.~\eqref{log-condizione}.

A useful technical tool in our work is the $d-$dimensional truncation of our entangled state $\ket{\phi}$, which can be obtained as $(\Pi_d \otimes \id) \ket{\phi}=(\id\otimes \Pi_d) \ket{\phi}=c_f \sum_{j=0}^{d-1} \phi_j \ket{j,j}$, where $\Pi_d=\sum_{j=0}^{d-1}\ketbra{j}$. The `approximation error' associated with this truncation can be quantified as follows:

 \begin{definition}\label{deftail}
 	The tail of our entangled state $\ket{\phi}$, dependent on the truncation dimension $d$, is defined as
 	\begin{equation}\label{tail}
 	\epsilon_d \coloneqq \big\|\left((\id-\Pi_d) \otimes \id\right) \ket{\phi} \big\| = c_f \sqrt{\sum_{j=d}^{\infty}\frac{1}{f_j}}
 	\end{equation}	
 \end{definition}
\noindent
The $f$-dependent entangled state $\ket{\phi}$ allows us to consider a modified version of the Choi--Jamio\l kowski states~\cite{WatrousNotes2011,emergenza} (\textit{$f$-Choi states} for brevity -- see below), that will be crucial to prove Theorem~\ref{main_th_f}:
\begin{definition}\label{defchoi}
	The \textit{modified Choi--Jamio\l kowski state} of a cptp map $\Lambda : \mathcal{D}(A') \rightarrow \mathcal{D}(B)$, for a given sequence of Hamiltonian eigenvalues $f=\{f_j\}_j$, is defined as
	\begin{equation}
	J_f(\Lambda) \coloneqq \id_A \otimes \Lambda_{A'} [\ketbra{\phi}] ,
	\end{equation}
	where $\ket{\phi}$ is given in Eq.~\eqref{phi}.
\end{definition}
\noindent

Having introduced all the required ingredients, we can now state the following theorem.
\begin{tmain}\label{main_th_f}
Let $A$ be a quantum system equipped with a Hamiltonian $H_A$  which satisfies the Gibbs hypothesis and which, when written as in Eq.~\eqref{H-assumption}, also satisfies Eq.~\eqref{log-condizione}. Consider an arbitrary cptp map $\Lambda  : \mathcal{D}(A) \rightarrow \mathcal{D}(B_1 \otimes \ldots \otimes B_N)$, and define the effective dynamics from $\mathcal{D}(A)$ to $\mathcal{D}(B_j)$ as $\Lambda_j \coloneqq \Tr_{ B\backslash B_j} \circ \Lambda$. For an arbitrary number $0< \delta <1$, there exists a POVM $\{M_l\}_l$ and a set $S \subseteq \{1,...,N\}$, with $|S| \geq (1-\delta)N$, such that, for all $j \in S$ and for any integer truncation dimension $d\geq 0$, we have that
\beq
\|\Lambda_j - \cE_j  \|_{\diamond H,E} \leq \frac{\zeta}{\delta} ,
\label{almost m&p}
\eeq
where the measure-and-prepare channel $\cE_j$ is given by
\beq \label{m&p channel}
\cE_j(X) \coloneqq  \sum_l   \Tr (M_l X) \tau_{j,l}
\eeq
for some family of states $\tau_{j,l} \in \mathcal{D}(B_j)$, and
\beq \label{zeta_main}
\zeta = \kappa d \left( \frac{ E^2 \sigma}{N c_f^4} \right)^{1/3} + \frac{4E}{c_f^2}\, \epsilon_d ,
\eeq
where $c_f$ is the normalization factor introduced in Eq.~\eqref{phi}; $\epsilon_d$ is given in Definition~\ref{deftail}; $\sigma$ is defined by Eq.~\eqref{entropy} and $\kappa\coloneqq 3\left(16\ln(2)\right)^{1/3}$ is a universal constant.
\end{tmain}

The complete proof is detailed in~\ref{appendix:a}. In what follows we provide the key ideas behind it.

\begin{proof}[Outline of the proof of Theorem~\ref{main_th_f}]
We start by proving that the 1-norm of an operator $L$, given by the difference between two $f$-Choi states, can be bounded as follows: 
\beq \label{ineqL}
\|L\|_1 \leq 4 d^\frac{3}{2} \max_{\cM} \| \id \otimes \cM [L] \|_1 + 4 \epsilon_d\, .
\eeq
Here, $\cM$ is an arbitrary measurement, thought of as a quantum-to-classical channel, $d$ is the truncation dimension and $\epsilon_d$ is given in Definition~\ref{deftail}. We then show that the distance between two channels is bounded by that between their $f$-Choi states:
	\beq\| \Lambda_0 - \Lambda_1 \|_{\Diamond H,E} \leq \frac{E}{c_f^2} \| J_f(\Lambda_0) - J_f(\Lambda_1) \|_1.\eeq 
The key ingredient of the proof is a result (Lemma~\ref{cor:brandao1} in~\ref{appendix:a}) which introduces a set of quantum-to-classical channels $\{{\cal M}_j|j\in J\}$ acting on a subset $J$ of the environment fragments $B_1;, \dots, B_N$. Let $z$ be the outcome of such set of measurements, then the state $\bbE_z \rho_A^z \otimes \rho_{B_j}^z$ can be proved to be the modified Choi--Jamio\l kowski state of a measure-and-prepare channel $\cE_j$ with POVM  independent of $j\notin J$. The Lemma bounds the quantity 
	\beq \label{lhslemma}
\bbE_{j \notin J} \max_{{\cal M}_j} \left\| \id \otimes {\cal M}_j \left[ \rho_{AB_j} - \bbE_z \rho_A^z \otimes \rho_{B_j}^z \right] \right\|_1 
\eeq
through a function of the entropy for system $A$; in~\eqref{lhslemma}, the expectation value is with respect to the uniform distribution over $\{1,\ldots, N\}\setminus J$, and the maximum is taken over all quantum-to-classical channels.\\ Since $\rho_{AB_j}=J_f(\Lambda_j)$ and $\bbE_z \rho_A^z \otimes \rho_{B_j}^z=J_f(\cE_j)$, by combining Lemma~\ref{cor:brandao1} with the previous inequalities we find a bound for the quantity $	\bbE_{j \notin J}\| \Lambda_j - \cE_j  \|_{\diamond H,E}$. We then easly obtain 
 $\bbE_j \| \Lambda_j - \cE_j  \|_{\diamond H,E}\leq \zeta$, where the index $j$ has uniform probability distribution over $\{1,...,N\}$, and $\zeta$ is given by Eq.~\eqref{zeta_main}.\\
We conclude the proof by applying Markov's inequality. In fact, the statement of the theorem is equivalent to the following one:
\beq \text{P}\left(\|\Lambda_j - \cE_j  \|_{\diamond H,E} \ge \frac{\zeta}{\delta} \right) \le \delta. \eeq\end{proof}

The result of Theorem~\ref{main_th_f} can be interpreted as follows. Fixing $0<\delta<1$ and $E$, and letting the number of environmental fragments $N$ tend to infinity, we have that the dynamical maps connecting the system to each of the fragments become indistinguishable from measure-and-prepare channels. This statement is true for at least a fraction $1-\delta$ of the sub-environments. Moreover, the measure-and-prepare channels involved are all defined by the same POVM $\{M_l\}_l$. For $\delta\ll 1$ this means that almost all observers probing the system by intercepting fragments of the environment can at most acquire classical information about one and the same measurement $\{M_l\}_l$ -- i.e., objectivity of observables holds for such observers. 

To illustrate the application of the results derived in this section to concrete physical models, we now consider some relevant examples.

\paragraph*{Case $f_j=j^2$, with $j\ge 1$ (particle in a box).} A quantum particle of mass $m$ confined in a box of length $L$ has Hamiltonian eigenvalues $f_j= \gamma j^2$, where $\gamma$ is a constant given by $\gamma =\frac{\hbar^2\pi^2}{2 m L^2}$. Choosing units such that $\gamma=1$ we have that $f_j=j^2$, and Theorem~\ref{main_th_f} turns out to hold for
 \beq 
 \label{particlebox}
 	\zeta = \alpha \left( \frac{ \sigma d^3 E^2 }{N} \right)^{1/3} +\beta E~\sqrt{\psi^{(d)}(1)},
 \eeq
where $\psi^{(n)}(z)$ is the $n^{th}$ derivative of the digamma function $\psi(z)$, $\sigma \approx 2.4$ and $\alpha,\beta$ are universal constants: $\alpha\coloneqq(12 \pi^4 )^{\frac{1}{3}}$, $\beta\coloneqq\sqrt{ \frac{8 \pi^2}{3}}$. 
In Figure~\ref{fig:plotbox} we plot the objectivity bound $\frac{\zeta}{\delta}$ provided by a numerical optimisation of Eq.~\eqref{particlebox} over $d$, with $E=1$ and $\delta = 0.01$.
 
\begin{figure}
\centering
  \includegraphics[width=0.55\linewidth]{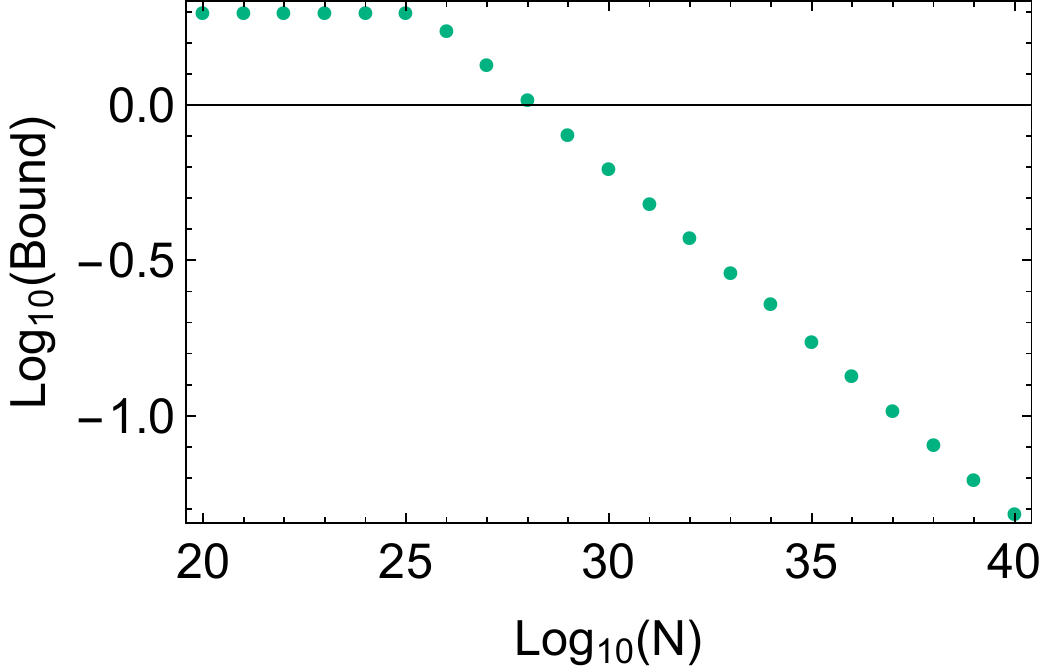}
  \caption{Case $f_j=j^2$. We plot the upper bound on $\| \Lambda_j - \cE_j \|_{\Diamond H,E}$ for $E=1$ and $\delta = 0.01$, obtained through numerical optimisation of Eq.~\eqref{particlebox} over the truncation dimension $d$.}
  \label{fig:plotbox}
\end{figure}
 
\paragraph*{Case of a $D$-dimensional system.} In this example, we show that our methods can bridge finite and infinite dimensions. Specifically, let us consider the sequence of Hamiltonian eigenvalues 
\begin{equation}
f_j=\left\{\begin{array}{ll}
1& j\leq D-1\,,\\
\frac{e^{\omega j}}{1-e^{-\omega}}& j\geq D\,,
\end{array}\right.
\end{equation}
where $D$ is a parameter that will turn out to be the actual Hilbert space dimension when $\omega\to\infty$. We assume $d\geq D$ for convenience. We obtain that
\begin{equation}\label{f_zeta}
\zeta = \left( {432 E^2(D+e^{-\omega D})^2 d^3 s \over N} \right)^{\!\!\frac13}\!\! + 4E \sqrt{\frac{D+e^{-\omega D}}{e^{\omega d}}}
\end{equation}
where 
\begin{align*}
s\coloneqq&\ \ln(2) \sigma\\
=&\ \ln\left(D\!+\!e^{-\omega D}\right)+{\omega\left(1\!-\!D\!+\!De^{\omega}\right)\over (e^\omega\!-\!1)(D e^{\omega D}\!+\!1)}-{\ln(1\!-\!e^{-\omega})\over(1\!+\!D e^{\omega D})}
\end{align*}
(see Example~\ref{bridge} in~\ref{appendix:a} for details). Taking the limit $\omega\to\infty$ we have that
\begin{align}\label{f_zeta_limit}
\lim_{\omega\to\infty}\zeta = \left( {432 E^2D^2 d^3 \ln D \over N} \right)^{1/3}.
\end{align}
In this scenario, $\Tr[\rho H]\leq E$ translates into the condition $\Tr[\rho]\leq1$, plus the additional constraint that the support of $\rho$ is contained in the $D$-dimensional subspace spanned by $\left\{\ket0,\ket1,...,\ket{D-1}\right\}$. Physically, the considered limit corresponds to raising all the Hamiltonian eigenvalues with $j\geq D$ to unattainably high energies, so that only levels with $j<D$ can be populated.

In the finite-dimensional scenario, the tightest objectivity bound to date has been recently obtained by Qi and Ranard~\cite{qi2020emergent}. The Qi--Ranard result can be compared to ours by making the substitution $|R|=1$, $|Q|= N\delta$ in Eq.~(12) of Ref.~\cite{qi2020emergent}. For the present comparison, we have to consider two possible expressions for the parameter $\Omega$ (which appears in Eq.~(13) of Ref.~\cite{qi2020emergent}): $\Omega=D^2$ and $\Omega=4D^{3/2}$, where $D$ is the dimension of system $A$; the corresponding expressions for the Qi--Ranard bound are the following:
	 \begin{align}\label{b1}
&\Omega = D^2 : &b_1 =  ~\left(\frac{2 D^6 \ln D}{ N \delta}\right)^{1/2} \\&\Omega = 4D^{3/2}: &b_2 = 4 ~\left(\frac{2 D^5 \ln D}{ N\delta}\right)^{1/2} \label{b2}
	\end{align}
Expressions $b_1$ and $b_2$ can be directly compared with our bound $b\coloneqq \frac{\zeta}{\delta}$, with $\zeta$ given by Eq.~\eqref{f_zeta_limit}, by taking $d=D$ (which clearly gives us the tightest bound for the range $d\geq D$) and by choosing $E=1$. The comparison is meaningful only in the regime in which the bounds are non-trivial, namely smaller than 2. By applying this requirement to our bound $b$ we obtain a threshold value for $N$ which depends on $D$: $N > \frac{54 D^5 \ln D}{ \delta^3}$. We find that our bound $b$ is never stronger than $b_2$ in its non-triviality regime. Note that $b_2$ was derived by the authors of~\cite{qi2020emergent} on the basis of our suggestion to exploit the Aubrun {\it et at.} result ~\cite{XOR} in this context. Hence $b_2$, which is the tightest known objectivity bound in finite dimensions, is a synthesis of independent insights from the analysis of Qi and Ranard and our work. The less tight bound $b_1$ was instead obtained by Qi and Ranard independently of the present work. There exists a regime in which our bound $b$ is non-trivial and stronger than $b_1$; however, this regime may be of little relevance in experimental contexts, as it requires a very large threshold for $N$, e.g.~rising above Avogadro's number for $\delta<0.1$.\smallskip

We now return to the case $f_j=j$, in which the condition $\sum{1\over f_j}<\infty$ is not satisfied. In this case the $f$-Choi states cannot be defined, and we replace them with truncated (standard) ones. To derive the objectivity bound we go through the same conceptual steps followed by Knott {\it et al.}\ in~\cite{emergenza}. However, we bound the distance between truncated Choi--Jamio\l kowski states more restrictively, by exploiting a result by Aubrun {\it et at.}~\cite[Corollary~9]{XOR}. We are then able to derive an objectivity bound that, for $d>16$, is tighther than the one obtained in~\cite{emergenza}. In particular, we find that Theorem~\ref{main_th_f} holds for $f_j=j$ with
 \beq
 \label{newnbound}
 	\zeta = \lambda \left( \frac{  d^5 \log(d)}{N} \right)^{1/3} + 4 \sqrt{\frac{E}{d}},
 \eeq
where $\lambda\coloneqq 3  (16 \ln(2))^{\frac{1}{3}}$. 
In Figure~\ref{fig:plotn} we compare the mean energy bound provided in~\cite{emergenza} (blue dots, uppermost curve) with the refined one we obtain from Eq.~\eqref{newnbound} (red dots, lowermost curve). Both bounds are numerically optimised over $d$ by setting $E=1$ and $\delta = 0.01$. 

\begin{figure}
\centering
  \includegraphics[width=0.55\linewidth]{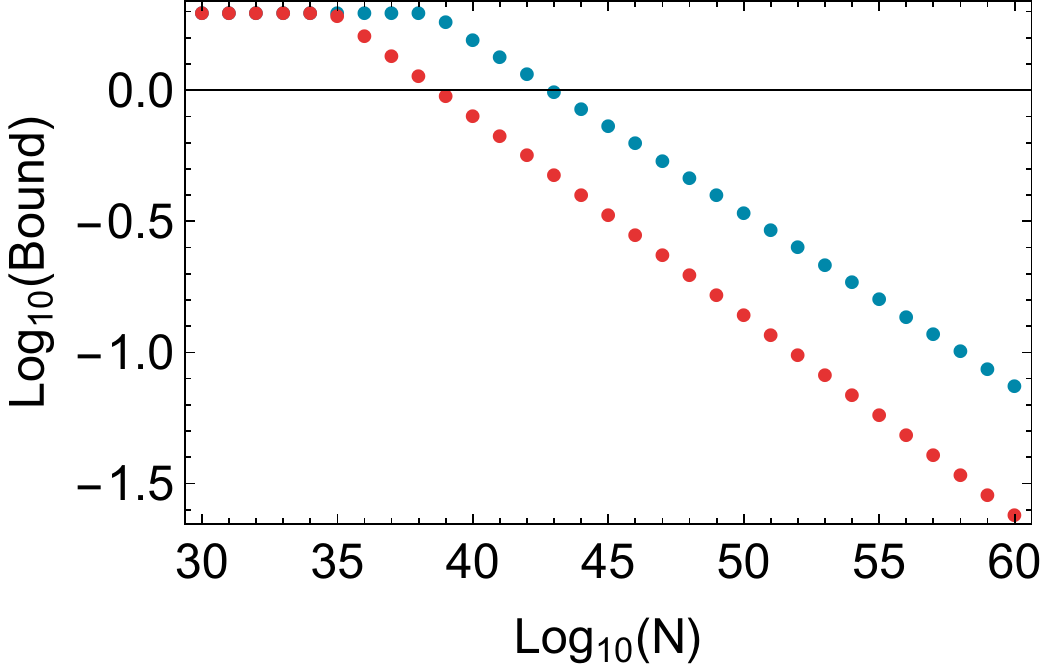}
  \caption{Case $f_j=j$. We compare the upper bound on $\| \Lambda_j - \cE_j \|_{\Diamond H,E}$ for $E=1$ and $\delta = 0.01$, obtained by numerical optimisation of Eq.~\eqref{newnbound} over $d$ (red dots, lowermost curve), with the bound obtained by Knott {\it et al.} in~\cite{emergenza}  (blue dots, uppermost curve).}
  \label{fig:plotn}
\end{figure}

\section{Testing optimality of the objectivity bound with an $\boldsymbol{N}$-splitter} \label{attenuator}
The emergence of objectivity of observables, as explored in the previous section as well as in~\cite{Brandao2015QD,emergenza}, is expressed by an upper bound on the distance between the effective dynamics $\Lambda_j$ and the measure-and-prepare channels $\cE_j$, which goes to zero as the number $N$ of environment fragments gets large. We now probe the optimality of such statement by looking at a lower bound for the distance between $\Lambda_j$ and $\cE_j$ in a specific example. This gives information on the speed at which the emergence of objectivity of observables takes place. We carry out this analysis for a system-environment interaction modelled by a pure loss channel. In detail, both our system $A$ and each of its sub-environments $B_1,...,B_N$ will be single bosonic modes with associated annihilation operators $a_0$ and $a_1,...,a_N$, respectively. The canonical commutation relations read $[a_j,a_k^\dagger]=\delta_{jk}$. We consider the quantum channel
\beq \label{channel}
\Lambda_{A\rightarrow B_1,...,B_N}(\cdot) \coloneqq U\left((\cdot)_A \otimes  \bigotimes_{j=2}^N|0\rangle \langle 0|_{B_j} \right)U^\dagger,
\eeq
where $U$ is the symplectic unitary which implements a $N-$splitter from $\mathcal{D}(A \otimes B_2 \otimes \ldots \otimes B_N)$ to $\mathcal{D}(B_1 \otimes \ldots \otimes B_N)$. In terms of bosonic operators (in the Heisenberg picture) this transformation takes the explicit form $U^{\dagger}a_l U = \sum_m V_{lm}a_m$, with $V_{lm}={1 \over \sqrt{N}}\exp{2 \pi i lm \over N}$. Since the initial environment state is the vacuum, the map in Eq.~\eqref{channel} corresponds to a pure loss channel of parameter $\frac{1}{N}$~\cite{HOLEVO-CHANNELS}. Varying the environment state one obtains instead a general attenuator~\cite{Koenig2015,Jack2018,KK-VV,G-dilatable,Lim2019,QCLT}. The reduced map $ \Lambda_j:\mathcal{D}(A) \rightarrow \mathcal{D}(B_j)$ is given by $ \Lambda_j = ~ \Tr_{B\backslash B_j}\circ \Lambda$ and has the same form for all $j$, as shown in~\ref{appendix:b1}. As in the previous section, we assume that system $A$ has bounded mean energy. As is typically the case in optical systems, the relevant Hamiltonian is obtained by setting $f_j=j$ (where $j$ may be interpreted as the number of photons). We show that, for a maximum energy threshold $E$ on system $A$ satisfying $E\ge \frac{2}{N}$, the channels $\Lambda_j$ approach the measure-and-prepare ones no faster than $\sim N^{-1}$. In particular, we can prove the following proposition.
\begin{proposition} \label{pure loss prop}
	Consider the cptp map $\Lambda  : \mathcal{D}(A) \rightarrow \mathcal{D}(B_1 \otimes \ldots \otimes B_N)$ given by Eq.~\eqref{channel}, and define $\Lambda_j \coloneqq \Tr_{B \backslash B_j} \circ \Lambda$ as the effective dynamics from $\mathcal{D}(A)$ to $\mathcal{D}(B_j)$. Let $E$ be the energy bound for system $A$, which is assumed to satisfy $E\ge \frac{2}{N}$. Then, for all POVMs $\{M_l\}_l$ and states $\{\tau_{j,l}\}_l \in \mathcal{D}(B_j)$, it holds that
	\beq \label{lower}
	\min_{j=1,\ldots,N} \left\|\Lambda_j - \cE_j  \right\|_{\diamond H,E} \geq  \dfrac{1}{2N},
	\eeq
	where the measure-and-prepare channel $\cE_j$ is given by Eq.~\eqref{m&p channel}.
\end{proposition}

\begin{remark}
The assumption $E\ge \frac{2}{N}$ in Proposition~\ref{pure loss prop} is not strictly necessary yet it significantly simplifies the calculation.
\end{remark}

\begin{proof}[Outline of the proof of Proposition~\ref{pure loss prop}]
We look at the quantity
\beq\mu(\Lambda) \coloneqq \inf_{M,\tau_{j}} 	\|\Lambda_j - \cE_{M,\tau_j}  \|_{\diamond H,E},\eeq
where the infimum is over the set of possible POVMs $M=\{M_l\}_l$ and states $\tau_j=\{\tau_{j,l}\}_l $ entering the definition of the measure-and-prepare channel $\cE_j$. We start by restricting the evaluation of the diamond norm to two-mode squeezed vacuum states:  $|\psi_r\rangle=\frac{1}{\cosh(r)} \sum_{n} \tanh(r)^n |nn\rangle$. Since the channels $\cE_{j}$ are entanglement-breaking, the infimum on $M$ and $\tau_{j}$ translates into an infimum on the set of separable states (with respect to the bipartition $C:B_j$, where $C$ is the ancillary system entering the definition of the diamond norm):
$(\id\otimes \cE_{j })[\psi_r] = \omega \in \mathrm{SEP}$, with $\psi_r\coloneqq|\psi_r\rangle \langle \psi_r|$. We thus obtain that
\beq \begin{split}\label{int}
\mu(\Lambda)\geq \inf_{\omega \in \mathrm{SEP}} \sup_{
	\substack{
		r:\, \sinh(r)^2 \leq E
	}
}   \|\id \otimes \Lambda_j  [\psi_r] - \omega \|_1\ .\end{split}\eeq 
A lower bound for the $1$-norm in Eq.~\eqref{int} is estimated through the inequality  $\| X\|_1 \geq 2 \| X \|_\infty$, which holds true for any operator $X$ with $\Tr X=0$. The operator norm on the r.h.s.\ is bounded from below by looking at the matrix entries with respect to a second set of two-mode squeezed vacuum states. Upon straightforward calculations, one obtains Eq.~\eqref{lower}.  Details are provided in~\ref{appendix:b2}.
\end{proof}

It is interesting to compare the lower bound in Eq.~\eqref{lower} with the upper bounds on the convergence rate we have found so far, with the goal of estimating the rate at which emergence of objectivity actually takes place.
To estimate an upper bound for the distance $\|\Lambda_j - \cE_j  \|_{\diamond H,E}$ we optimise Eq.~\eqref{newnbound} over $d$ by using the inequality $\ln(d)\le d$, and exploit the fact that, for the model we are considering, all the reduced maps $\Lambda_j$ have the same form.
We thus obtain the following range:
\beq  \dfrac{1}{2N}  \leq
\|\Lambda_j - \cE_j  \|_{\diamond H,E} \leq \mu \left(\frac{E^6}{N}\right)^{1 \over 15},
\eeq
where $\mu<10$ is a constant.



\section{Quantum discord from local redistribution of quantum correlations in infinite dimension}
\label{discord}
Quantum discord~\cite{ollivier2001quantum,henderson2001classical} is regarded as a measure of the purely quantum part of correlations between systems~\cite{modi2012classical,ABC}. Consider two systems $A$ and $B$, collectively described by a state $\rho$; the total amount of correlations between them is quantified by the mutual information $I(A\!:\!B) = S(A)+S(B)-S(AB)$, where $S$ denotes the von Neumann entropy: $S(A) = -\Tr \left[\rho_A \log \rho_A \right]$. The quantum discord between $A$ and $B$ (from the perspective of subsystem $B$) is then defined by
\beq \label{defdiscord}
D(A|B)_{\rho}\coloneqq I(A\!:\!B)_{\rho} - \max_{\Gamma \in QC} I(A\!:\!B)_{(\id\! \otimes \Gamma)(\rho)} , 	
\eeq
where $QC$ refers to quantum-to-classical channels having the form $\Gamma(X) \coloneqq  \sum_k \Tr [N_k X] \ketbra{k}$, with POVM $\{N_k\}_k$. 
The quantum discord $D(A|B)_{\rho}$ thus represents the amount of correlations that is inevitably lost when $B$ is subject to a minimally disturbing local measurement, or, in other words, when $B$ encodes its part of information into a classical system; in this respect, $D(A|B)_{\rho}$ can be thought of as the purely quantum part of correlations between $A$ and $B$ in the state $\rho$. In~\cite{Brandao2015QD} Brand\~{a}o {\it et al.}\ derived an interesting operational interpretation of quantum discord in terms of redistribution of quantum information to many parties. In particular they showed that
\beq \label{Bcor}
\lim_{N \to \infty} \max_{\Lambda_{\!N\!}}\, \bbE_j I(A\!:\!B_j)_{(\id\!\otimes \Lambda_{\!N\!})(\rho)} = \max_{\Gamma \in QC} I(A\!:\!B)_{(\id \!\otimes \Gamma)(\rho)}
\eeq 
where the maximisation is over all maps $\Lambda_N:\mathcal{D}(B) \rightarrow \mathcal{D}(B_1 \otimes \ldots \otimes B_N)$, and $\bbE_j I(A\!:\!B_j)$ is the average mutual information between $A$ and $B_j$ for the uniform probability distribution over $j$. Equation~\eqref{Bcor} shows that, when the share of correlations of $B$ is  redistributed to infinitely many parties $\{B_j\}$, the maximum \textit{average} mutual information accessible through each one of the parties $B_j$ corresponds to the purely classical part of correlations. This result is at the heart of the operational characterisation of quantum discord provided by Brand\~{a}o {\it et al.}~\cite{Brandao2015QD}. In fact, from Eq.~\eqref{Bcor} it follows that
\beq\label{inter}
D(A|B)_{\!\rho
} = \lim_{N \to \infty} \min_{\Lambda_{\!N\!}}\, \bbE_j\! \left( I(A\!:\!B)_{\rho}\! -\! I(A\!:\!B_j)_{(\id\! \otimes \Lambda_{\!N\!}) (\rho)} \right),
\eeq
i.e., $D(A|B)_{\rho}$ is characterised as the minimal average loss in mutual information when $B$ locally redistributes its share of correlations. Brand\~{a}o {\it et al.}~derived Eq.~\eqref{Bcor} as a corollary of the theorem through which they proved emergence of objectivity of observables in finite dimensions ~\cite[Corollary~4]{Brandao2015QD}. In that context, $B$ can be interpreted as the environment of system $A$, which splits into fragments $\{B_j\}$. 

We generalise the above result to an infinite-dimensional scenario, for systems subjected to generic energy constraints. In particular, as infinite-dimensional counterpart of \cite[Corollary~4]{Brandao2015QD}, we prove the following corollary of our Theorem~\ref{main_th_f}:

\begin{corollary}\label{cor}
Let $A$ be a quantum system equipped with a Hamiltonian $H_A$, and $B$ a quantum system equipped with Hamiltonian $H_B$, both satisfying the Gibbs hypothesis. We also assume that $H_B$, when written as in Eq.~\eqref{H-assumption}, satisfies Eq.~\eqref{log-condizione}. Let $\Lambda_N  : \mathcal{D}(B) \rightarrow \mathcal{D}(B_1 \otimes \ldots \otimes B_N)$ be a cptp map, and define $\Lambda_j \coloneqq  \Tr_{B \backslash B_j} \circ \Lambda_N$ as the effective dynamics from $\mathcal{D}(B)$ to $\mathcal{D}(B_j)$. Then for every $\delta > 0$ there exists a set $S \subseteq \{1,...,N\}$ with $|S| \geq (1-\delta)N$ such that for all $j \in S$ and all states  $\rho \in \mathcal{D}(A \otimes B)$ with $\Tr [\rho H_A] \leq E_A, \Tr[\rho_B H_B] \leq E_B$, 
	\beq \label{dis}
	\begin{split}
	I(A:B_j)_{(\id \otimes \Lambda_j) (\rho)} \leq &\max_{\Gamma \in QC} 	I(A:B)_{(\id \otimes \Gamma) (\rho)} \\&+ (2 \epsilon' + 4\Delta) S(\gamma(E_A/\Delta)) \\&+ (1 + \epsilon') \,h\left({\epsilon' \over {1 + \epsilon'} }\right) + 2 h(\Delta)\ ,\end{split}
	\eeq
	where $\epsilon' = {\zeta \over \delta} $, $\Delta = \frac{1}{2}{\epsilon' \over 1 + \epsilon'}$, $\gamma(E)$ is the Gibbs state for system $A$ defined in Eq.~\eqref{gammamain}, and the maximum on the r.h.s.\ is over quantum-to-classical channels $\Gamma(X) \coloneqq \sum_l \Tr (N_l X) \proj{l}$,  with $\{N_l\}_l$ a POVM and  $\{\ket{l}\}_l$ a set of orthonormal states. As a consequence,
\beq \label{rs}\lim_{N \to \infty} \max_{\Lambda_{\!N\!}}\, \bbE_j I(A\!:\!B_j)_{(\id\!\otimes \Lambda_{\!N\!})(\rho)} = \max_{\Gamma \in QC} I(A\!:\!B)_{(\id \!\otimes \Gamma)(\rho)}\eeq
\end{corollary}
\begin{proof}[Outline of the proof of Corollary~\ref{cor}] We follow the conceptual steps of the proof of~\cite[Corollary~4]{Brandao2015QD}, adapting them to our infinite-dimensional framework. In particular, our argument relies on a continuity bound for the conditional entropy of infinite-dimensional systems subjected to energy constraints~\cite[Lemma~17]{tightuniform}. We apply it to the states $\tau = (\id \otimes \Lambda_j)(\rho)$ and $\sigma= (\id \otimes \cE_j)(\rho)$, which are close in 1-norm by virtue of Theorem~\ref{main_th_f}. Since the reduced entropies on the $A$ subsystems are the same for $\tau$ and $\sigma$, the continuity bound for the conditional entropy holds true for the mutual information as well. We then obtain Eq.~\eqref{dis}. To prove Eq.~\eqref{rs}, we exploit Eq.~\eqref{dis} to show that the l.h.s.~is no larger that the r.h.s.; this concludes the proof, as the reverse (r.h.s.~no larger than l.h.s.) is trivial. The complete proof is given in~\ref{appendix:c}. \end{proof}\begin{remark}
The result of Corollary~\ref{cor} also applies to a Hamiltonian $H_B$ that takes the form~\eqref{H-assumption} with $f_j=j$, and therefore does not satisfy Eq.~\eqref{log-condizione}. In fact, the proof remains valid when the objectivity bound of Theorem~\ref{main_th_f} is replaced with the one given by~\eqref{newnbound}.\end{remark}
 As mentioned before, Eq.~\eqref{rs} implies that quantum discord can be interpreted as the minimal average loss in mutual information when one of the two parties asymptotically redistributes its share of correlations. In the framework of Quantum Darwinism this means that, when the number of environment fragments grows significantly, the correlations established between the (infinite-dimensional) system of interest $A$ and each of the observers (who in turn has access only to a fragment $B_j$ of the environment) can be at most classical.

\section{Conclusions and Outlook}

In this paper we  investigated the generic characteristics of the objectivity of observables arising in the quantum-to-classical transition within the premises of Quantum Darwinism. Going beyond recent studies for finite- and infinite-dimensional systems ~\cite{Brandao2015QD,emergenza}, we presented a unified approach to derive bounds on the emergence of such objectivity in quantum systems of arbitrary dimension, probed by multiple observers each accessing a fragment of the environment. In the particular case of a system-environment dynamics specified by a pure loss channel, we derived lower and upper bounds on the rate at which objectivity of observables emerges as a function of the number of environmental fragments. Furthermore, we proved that, even when the system under observation is infinite-dimensional, it cannot share quantum correlations with asymptotically many observers, as the maximum correlation each observer can establish with the system is, on average, of purely classical nature. This observation, which extends to the infinite-dimensional scenario an operational interpretation for quantum discord put forward in~\cite{Brandao2015QD}, can also be seen as a quantitative manifestation of the quantum-to-classical transition, seen exclusively from the balance of correlations, without having to analyse the system-environment interaction. 

The role of quantum discord in understanding the quantum-to-classical transition has also been recently investigated in Ref.~\cite{TK}. In particular, the authors showed the equivalence between the so-called {\em strong Quantum Darwinism} and {\em spectrum broadcasting}, another framework aiming for the modelling and interpretation of  ``objectivity"~\cite{Horodecki2015}. Exploring deeper connections between these studies and our results, with the aim to achieve an even more fundamental (and quantitative) understanding of the emergence of objectivity and classicality, is certainly an endeavour worthy of further investigation. Another fascinating perspective could be to study the applicability of our methods --- which are rooted in quantum information theory and related, e.g., to no-broadcasting and monogamy properties of genuinely quantum correlations --- to cosmological scenarios~\cite{ErBucodeZurek}, in order to cast new light on the black hole information paradox and related issues remaining unsolved at the quantum/classical/general-relativistic triple border.

\section*{Acknowledgments}
We acknowledge financial support from the European Research Council (ERC) under the Starting Grant GQCOP (Grant no.~637352) and the Foundational Questions Institute (FQXi) under the Intelligence in the Physical World Programme (Grant no.~FQXiRFP-IPW-1907). LL is supported by the University of Ulm. TT acknowledges support from the University of Nottingham via a Nottingham Research
Fellowship. We are grateful to Marco Piani and Paul Knott for insightful discussions.

\section*{References}

\bibliographystyle{iopart-num-g}
\bibliography{bibliografia}

\appendix 

\renewcommand*{\thetheorem}{A\arabic{theorem}}
\renewcommand*{\thetmain}{M\arabic{tmain}}
\renewcommand*{\theexample}{A\arabic{example}}
\renewcommand*{\thelemma}{A\arabic{lemma}}
\renewcommand*{\thecorollary}{A\arabic{corollary}}
\renewcommand*{\theremark}{A\arabic{remark}}
\renewcommand*{\theobservation}{A\arabic{observation}}
\renewcommand*{\theproposition}{A\arabic{proposition}}
\renewcommand*{\thedefinition}{A\arabic{definition}}

\section{$\boldsymbol{f}$-dependent objectivity bounds}
\label{appendix:a}

The proof of Theorem~\ref{main_th_f} involves a generalisation of the concept of Choi--Jamio\l kowski isomorphism, which relies on a class of infinite-dimensional entangled states depending on the underlying system's Hamiltonian. For clarity, we recall here the definition of modified Choi--Jamio\l kowski state associated to such class of states. 

\begin{manualdefinition}{\ref{defchoi}}[(Restatement)]
The \textit{modified Choi--Jamio\l kowski state} of a cptp map $\Lambda: \mathcal{D}(A') \rightarrow \mathcal{D}(B)$, for a given sequence of Hamiltonian eigenvalues $f=\{f_j\}_j$, is defined as
	\begin{equation}
	J_f(\Lambda) \coloneqq \id_A \otimes \Lambda_{A'} [ \ket{\phi}\bra{\phi} ],
	\end{equation}
	where the entangled state $\ket{\phi}$ reads
	\beq \tag{\ref{phi}}
	\ket{\phi} \coloneqq c_f \sum_{j=0}^{\infty} \phi_j \ket{j,j}_{AA'}\,,
	\eeq
	with $\phi_j^2 \coloneqq 1 / f_j$ and $c_f \coloneqq \left(\sum \frac{1}{f_j}\right)^{-\frac12}$. 
\end{manualdefinition}

We start by proving Lemma~\ref{trunc_1}, which bounds the distance between two $f$-Choi states as a function of $d$ (the truncation dimension), and Lemma~\ref{LJ}, which relates the distance between two channels to that between their $f$-Choi states. Our preparations are completed by the rather technical Lemma~\ref{cor:brandao1}: there we show that a crucial inequality exploited in Ref.~\cite{Brandao2015QD}, whose original formulation explicitly relies on finite-dimensional techniques, may be suitably modified to fit our infinite-dimensional scenario. To achieve the latter result, we exploit the assumption that the $f$-Choi states have finite local entropy.
Once all the above ingredients are in place, we present the proof of Theorem~\ref{main_th_f}. We conclude the section by presenting additional details on the calculations behind Eq.~\eqref{f_zeta}, obtained for the sequence $\{f_j\}$ which bridges the finite- and infinite-dimensional cases. \smallskip

Lemma~\ref{trunc_1} is a generalised and refined version of Lemma~S5 in the supplemental material of~\cite{emergenza}. The first property comes from our definition of modified Choi--Jamio{\l}kowski states, which relies on generic Hamiltonian eigenvalues $\{f_j\}$ (whilst in~\cite{emergenza} the latter take an exponential form). The second one arises from applying a result by Aubrun {\it et at.}~\cite[Corollary~9]{XOR}.
\begin{lemma}\label{trunc_1}
	Given $L = \tau - \sigma $, where $\tau=J_f(\Lambda_1)$ and $\sigma=J_f(\Lambda_2)$ are modified Choi--Jamio{\l}kowski states for the cptp maps $\Lambda_1$ and $\Lambda_2$, we have that
	\begin{align}
	\|L\|_1 \leq 4 d^\frac{3}{2} \max_{\cM} \| \id \otimes \cM[L] \|_1 + 4 \epsilon_d\ ,
	\end{align}
	where $\epsilon_d$ is given in Definition~\ref{deftail}, $d$ is the corresponding truncation dimension, and the maximum on the r.h.s.\  is over quantum-to-classical channels $\cM(Y) =
	\sum_l \Tr(N_l Y) \ket{l}\bra{l}$, with POVM $\{N_l\}_l$ and orthonormal states $\{\ket{l}\}_l$.
\end{lemma}

\begin{proof}
By writing $L$ in the form $L = \sum_{ij =0}^{\infty} \ket{i}\bra{j} \otimes L_{ij}$ we have that
	\begin{align}
	\|L\|_1 &\overset{1}{\leq}	\| (\Pi_d \otimes \id) [L] \|_1 +
	\| ((\id - \Pi_d) \otimes \id) [L] \|_1  \\ \nonumber
	&= \| (\Pi_d \otimes \id) [L] \|_1 +
	\left\|\sum_{\min\{i, j\} \geq d} \ket{i}\bra{j}\otimes L_{ij} \right\|_1  \\ \nonumber
	&\overset{2}{\leq} 4 d^\frac{3}{2} \max_{\cM} \left\| (\Pi_{d} \otimes \cM) [L] \right\|_1 + \left\|\sum_{\min\{i, j\} \geq d} \ket{i}\bra{j}\otimes L_{ij} \right\|_1 \\ \nonumber
	&\overset{3}{\leq} 4 d^\frac{3}{2} \max_{\cM} \left\| (\id \otimes \cM) [L] \right\|_1 + \left\|\sum_{\min\{i, j\} \geq d} \ket{i}\bra{j}\otimes L_{ij} \right\|_1 .
	\end{align}
Note that in~1 we used the triangle inequality. In~2, instead, we applied a result by Aubrun {\it et al.}~\cite[Corollary~9]{XOR}: in fact, for an arbitrary bipartite operator $Z$, it holds that
\begin{equation*}
\max_{\cM} \left\| (I \otimes \cM) [Z] \right\|_1 = \|Z\|_{\mathrm{LOCC}_\leftarrow} \geq \|Z\|_{\mathrm{LO}}\, ,
\end{equation*}
where the maximization on the l.h.s.\ is as usual over local measurements, while the quantities (a)~$\|\cdot\|_{\mathrm{LOCC}_\leftarrow}$ and (b)~$\|\cdot\|_{\mathrm{LO}}$ are the distinguishability norms~\cite{VV-dh,ultimate} associated with the sets of (a)~local operations assisted by classical communication from the second system to the first; and (b)~local operations alone. By~\cite[Eq.~(40)]{XOR}, it holds that $\|Z\|_{\mathrm{LO}}\geq \frac{1}{4n^{3/2}}\|Z\|_1$, where $n$ denotes the smaller of the local dimensions. The inequality in~2 is just an application of this, with $Z\coloneqq (\Pi_d \otimes \id) [L]$ and hence $n\leq d$.
Furthermore, in~3 we applied the pinching theorem~\cite[Eq.~(IV.52)]{BHATIA-MATRIX}, or, alternatively, the data processing inequality for the trace distance -- note that $X\mapsto \Pi X\Pi + (\id-\Pi)X (\id-\Pi)$ is a cptp map for every projector $\Pi$. Finally, multiple applications of the triangle inequality yield
	\begin{align*}
	\left\|\sum_{\min\{i,j\} \geq d} \ket{i}\bra{j} \otimes L_{ij} \right\|_1
	&= \left\|L - (\Pi_{d} \otimes \id) L (\Pi_{d} \otimes \id) \right\|_1 \\
	&= \left\|(\tau-\sigma) - (\Pi_{d} \otimes \id) (\tau-\sigma) (\Pi_{d} \otimes \id) \right\|_1 \\
	&= \left\|\tau - \tau_{d} - ( \tau - \sigma_{d} ) \right\|_1 \nonumber \\
	&\leq \left\|\tau - \tau_{d} \right\|_1 + \left\| \sigma - \sigma_{d} \right\|_1 
	\end{align*}
	where $\tau_{d}\coloneqq (\Pi_d \otimes \id) \tau (\Pi_d \otimes \id)$ and $\sigma_d\coloneqq (\Pi_d \otimes \id) \sigma (\Pi_d \otimes \id)$.
The $1-$norms on the r.h.s.\  can be bounded from above by exploiting the result of Proposition~S2 in the supplemental material of \cite{emergenza}, suitably adapted to our  modified Choi--Jamio{\l}kowski states. In particular, by replacing the coefficients $\phi_j = e^{-\frac{\omega j}{2}}$ in
\cite[Proposition~S2]{emergenza} with our $\phi_j = {f_j}^{-{1\over 2}}$ we have that, for $\rho=J_f(\Lambda)$,
	\begin{equation}
	\|\rho - \rho_d \|_1 \leq 2\epsilon_d\ ,
	\end{equation} where $\rho_d \coloneqq (\Pi_d \otimes \id) \rho (\Pi_d \otimes \id)$. We then obtain  
		\begin{align*}
	\left\|\sum_{\min\{i,j\} \geq d} \ket{i}\bra{j} \otimes L_{ij} \right\|_1
\leq 4\epsilon_d\ .\nonumber
	\end{align*}
\end{proof}

Before proceeding with the proof, we restate for clarity the definition of energy-constrained diamond norm (see Definition~\ref{diamondH} in the main text) for the specific case of a Hamiltonian which satisfies the Gibbs hypothesis and is written as in Eq.~\eqref{H-assumption}.

\begin{definition} \label{diamondH-f}
Let $A'$ be a quantum system equipped with a Hamiltonian $H_{A'}$ satisfying the Gibbs hypothesis and written as in Eq.~\eqref{H-assumption}, and pick $E > E_0 = f_0$. Then the energy-constrained diamond norm of an arbitrary Hermiticity-preserving linear map $\Lambda : \mathcal{D}(A') \rightarrow \mathcal{D}(B)$ is defined by
	\beq \label{diamondapp}
	\|\Lambda\|_{\Diamond H,E}\coloneqq
	\sup_{
		\substack{
			\sum_j f_j \braket{j|\rho_{A'}|j} \leq E
		}
	} \left\| (\id_A \otimes \Lambda_{A'})(\rho_{AA'})\right\|_1 ,
	\eeq
where $A$ is an arbitrary ancillary system, and $ \| \cdot \|_1$ is the one-norm. A recent result by Weis and Shirokov~\cite{Weis-Shirokov} ensures that the input state $\rho_{AA'}$ in Eq.~\eqref{diamondapp} can be taken to be pure.
\end{definition}

\begin{lemma}[(Generalisation of Lemma S6 in the supplemental material of~\cite{emergenza} for a Hamiltonian given by Eq.~\eqref{H-assumption} and Eq.~\eqref{log-condizione})] \label{LJ}
For cptp maps $ \Lambda_0 $ and $ \Lambda_1$ whose input system is equipped with a Hamiltonian $H$ which satisfies the Gibbs hypothesis, takes the form as in Eq.~\eqref{H-assumption} and satisfies Eq.~\eqref{log-condizione}, we have that
\beq
\| \Lambda_0 - \Lambda_1 \|_{\Diamond H,E} \leq \frac{E}{c_f^2} \| J_f(\Lambda_0) - J_f(\Lambda_1) \|_1\ ,
\eeq
where the modified Choi--Jamio\l kowski state $J_f(\Lambda)$ of $\Lambda$ is constructed as in Definition~\ref{defchoi}.
\end{lemma}

\begin{proof}
Lemma~\ref{LJ} can be proved by adapting the argument in the proof of \cite[Lemma~S6]{emergenza} to our choice of the imput system's Hamiltonian, i.e., by replacing the definition of modified Choi--Jamio\l kowski states used there with the one given in Definition~\ref{defchoi}.
\end{proof}

\begin{lemma}[(Adapted from Eq.~(16) in the supplementary notes of~\cite{Brandao2015QD})] \label{cor:brandao1}
	Let $\Lambda$ be a cptp map, and let the corresponding modified Choi--Jamio\l kowski state given by Definition~\ref{defchoi} be denoted with $\rho_{AB_1...B_N}\coloneqq \id_{A} \otimes \Lambda_{A'} (\ket{\phi}\bra{\phi})$, where $\ket{\phi}$ is given in Eq.~\eqref{phi}.
	Fix an integer $m\leq N$. Then there exists a set of indices $J\coloneqq \left(j_1,\ldots, j_{q-1}\right)$, where $q\le m$, and quantum-to-classical channels ${\cal M}_{j_1},\ldots, {\cal M}_{j_{q-1}}$ such that 
	\beq \label{step1}
	\bbE_{j \notin J} \max_{{\cal M}_j} \left\| \id \otimes {\cal M}_j \left[ \rho_{AB_j} - \bbE_z \rho_A^z \otimes \rho_{B_j}^z \right] \right\|_1 \leq \sqrt{2 \ln(2) \sigma \over m},
	\eeq
	where: $\sigma$ is given in~Eq.\eqref{entropy}; the expectation value is with respect to the uniform distribution over $\{1,\ldots, N\}\setminus J$; the maximum runs over all quantum-to-classical channels; $z$ is a random variable that represents the outcome of the measurements ${\cal M}_{j_1},\ldots, {\cal M}_{j_{q-1}}$ on $\rho_{AB_1\ldots B_N}$; and $\rho_A^z$, $\rho_{B_j}^z$ are the corresponding post-measurement states.
	\end{lemma}

\begin{proof}
It suffices to adapt the derivation of Eq.$(16)$ in the supplementary notes of~\cite{Brandao2015QD} to our infinite-dimensional scenario: the Choi--Jamio\l kowski state of $\Lambda$ is replaced with the $f-$Choi state of Definition~\ref{defchoi}, and the entropy $\log d_A$ with $S(\rho_A)$. Since $\Lambda$ is trace preserving, $\rho_A = \Tr_{A'}[\ketbra{\phi}_{AA'}]$, and $S(\rho_A)=\sigma$ by definition of $\sigma$.
\end{proof}


\begin{manualtmain}{\ref{main_th_f}}[(Restatement)] Let $A$ be a quantum system equipped with a Hamiltonian $H_A$  which satisfies the Gibbs hypothesis and which, when written as in Eq.~\eqref{H-assumption}, also satisfies Eq.~\eqref{log-condizione}. Consider an arbitrary cptp map $\Lambda  : \mathcal{D}(A) \rightarrow \mathcal{D}(B_1 \otimes \ldots \otimes B_N)$, and define the effective dynamics from $\mathcal{D}(A)$ to $\mathcal{D}(B_j)$ as $\Lambda_j \coloneqq \Tr_{ B\backslash B_j} \circ \Lambda$. For an arbitrary number $0< \delta <1$, there exists a POVM $\{M_l\}_l$ and a set $S \subseteq \{1,...,N\}$, with $|S| \geq (1-\delta)N$, such that, for all $j \in S$ and for any integer truncation dimension $d\geq 0$, we have that
\beq
\|\Lambda_j - \cE_j  \|_{\diamond H,E} \leq \frac{\zeta}{\delta} ,
\tag{\ref{almost m&p}}
\eeq
where the measure-and-prepare channel $\cE_j$ is given by
\beq \tag{\ref{m&p channel}}
\cE_j(X) \coloneqq  \sum_l   \Tr (M_l X) \tau_{j,l}
\eeq
for some family of states $\tau_{j,l} \in \mathcal{D}(B_j)$, and
\beq \tag{\ref{zeta_main}}
\zeta = \kappa d \left( \frac{ E^2 \sigma}{N c_f^4} \right)^{1/3} + \frac{4E}{c_f^2}\, \epsilon_d ,
\eeq
where $c_f$ is the normalization factor introduced in Eq.~\eqref{phi}; $\epsilon_d$ is given in Definition~\ref{deftail}; $\sigma$ is defined by Eq.~\eqref{entropy} and $\kappa\coloneqq 3\left(16\ln(2)\right)^{1/3}$ is a universal constant.
\end{manualtmain}

\begin{proof}
	As we will see below, the states $\rho_{AB_j}$ and $\bbE_z \rho_A^z \otimes \rho_{B_j}^z$ defined in Lemma~\ref{cor:brandao1} are modified Choi--Jamio\l kowski states. By applying Lemma~\ref{trunc_1} to them we have that
	\begin{align}\label{step2}
	\left\|\rho_{AB_j} - \bbE_z \rho_A^z \otimes \rho_{B_j}^z\right\|_1 \leq 4 d^\frac{3}{2} \max_{{\cal M}_j} \left\| \id_A \otimes {\cal M}_j \left[ \rho_{AB_j} - \bbE_z \rho_A^z \otimes \rho_{B_j}^z \right] \right\|_1 + 4\epsilon_d.
	\end{align}
	By combining Eq.~\eqref{step2} with Lemma~\ref{cor:brandao1} we then obtain
	\begin{equation}
	\begin{aligned}
	\bbE_{j \notin J} \left\|\rho_{AB_j} - \bbE_z \rho_A^z \otimes \rho_{B_j}^z \right\|_1	&\leq 4 d^\frac{3}{2} \bbE_{j \notin J} \max_{{\cal M}_j} \left\| \id_A \otimes {\cal M}_j \left[ \rho_{AB_j} - \bbE_z \rho_A^z \otimes \rho_{B_j}^z \right] \right\|_1 + 4\epsilon_d \\
	&\leq 4 d^\frac{3}{2} \sqrt{2 \ln(2) {\sigma} \over m} + 4\epsilon_d . 
	\end{aligned}
	\label{18-l}
	\end{equation}

We now show that $\bbE_z \rho_A^z \otimes \rho_{B_j}^z$ is the modified Choi--Jamio\l kowski state of a quantum-to-classical channel, explicitly given by
	\beq\label{mp1}
	\cE_j(X) \coloneqq c_f^{-2} \bbE_z \Tr \left[(\rho_A^z)^\intercal H^{1\over2} X H^{1\over2}\right] \rho_{B_j}^z .
	\eeq
In fact, 
\begin{align*}
    \left(\id_A\otimes \cE_j \right)\left(\ket{\phi}\bra{\phi}\right) &= c_f^{2} \sum_{j,k} \frac{1}{{f_j}^{1\over2}}\frac{1}{{f_k}^{1\over2}}\ket{j}\bra{k}\otimes \cE_j\left( \ket{j}\bra{k} \right) \\
    &= \bbE_z \sum_{j,k} \braket{j|\rho_A^z|k} \ket{j}\bra{k} \otimes \rho_{B_j}^z \\
    &= \bbE_z \rho_A^z\otimes \rho_{B_j}^z .
\end{align*}
Note that the measurement appearing in Eq.~\eqref{mp1} is independent of $j\notin J$. In fact, calling $N^z_{B_{j_1}\ldots B_{j_{q-1}}}$ the POVM element corresponding to the outcome $z$ of the measurement ${\cal M}_{j_1}\otimes \ldots \otimes {\cal M}_{j_{q-1}}$, the POVM appearing in Eq.~\eqref{mp1} can be expressed as $\left\{c_f^{-2} p(z)\, H^{1\over2} (\rho_A^z)^\intercal H^{1\over2} \right\}_z$, where $p(z) = \Tr\left[\rho_{AB_1\ldots B_N} N^z_{B_{j_1}\ldots B_{j_{q-1}}} \right]$. Now the claim follows because
\begin{equation}
    p(z) \left( \rho_A^z \right)^\intercal = \Tr_{B_1\ldots B_N}\left[ \rho_{AB_1\ldots B_N} N^z_{B_{j_1}\ldots B_{j_{q-1}}} \right]
\end{equation}
is independent of $j\notin J$.

\noindent
Since $\rho_{AB_j}$ is, by definition, the modified Choi--Jamio\l kowski state of $\Lambda_j$, from Lemma~\ref{LJ} it follows that
	\beq\label{20-l}
	\left\| \Lambda_j - \cE_j  \right\|_{\diamond H,E} \leq \frac{E}{c_f^2} \left\| \rho_{AB_j} - \bbE_z \rho_A^z \otimes \rho_{B_j}^z \right\|_1.
	\eeq
This, combined with Eq.~\eqref{18-l}, gives
	\begin{align*}
	\bbE_{j \notin J} \| \Lambda_j - \cE_j  \|_{\diamond H,E} &\leq \frac{E}{c_f^2} \bbE_{j \notin J} \| \rho_{AB_j} - \bbE_z \rho_A^z \otimes \rho_{B_j}^z \|_1 \\
	&\leq \frac{E}{c_f^2} \left( 4 d^\frac{3}{2} \sqrt{2 \ln(2) {\sigma} \over m} + 4\epsilon_d \right) \\
	&= \sqrt{\frac{32\ln(2) E^2 d^3 \sigma}{m c_f^4}} + \frac{4E}{c_f^2} \epsilon_d\ .
	\end{align*}
From the previous result we then find that
	\begin{align*}
	\bbE_j \| \Lambda_j - \cE  \|_{\diamond H,E} &\leq \bbE_{j \notin J} \| \Lambda_j - \cE_j  \|_{\diamond f} + {m \over N} \bbE_{j \in J} \| \Lambda_j - \cE_j  \|_{\diamond f} \\
	&\leq \sqrt{\frac{32\ln(2) E^2 d^3 \sigma}{m c_f^4}} + \frac{4E}{c_f^2} \epsilon_d + {2m \over N}.
	\end{align*}
	The right-hand-side, minimised with respect to $m$, gives the quantity
	\begin{align}\label{exp_zeta}
	\zeta =\kappa d\left( \frac{E^2\sigma}{N c_f^2} \right)^{1/3} + \frac{4E}{c_f^2} \epsilon_d,
	\end{align}
	where $\kappa= 3\left(16\ln(2)\right)^{1/3}$. To complete the proof, we apply Markov's inequality: $P \left(X\ge a \right) \le \frac{\bbE(X)}{a}$, where $X$ is a non-negative random variable, $\bbE(X)$ its expectation value, and $a >0$. In our case, $X =\left\| \Lambda_j - \cE  \right\|_{\diamond H,E}$, with $j$ being uniformly distributed, and $a=\frac{\zeta}{\delta}$, which leads us to
    \beq
    \text{P}\left(\|\Lambda_j - \cE_j  \|_{\diamond H,E} \ge \frac{\zeta}{\delta} \right) \le \delta ,
    \eeq
    completing the proof.
\end{proof}

\begin{example}[Case study: bridging finite and infinite dimensions] \label{bridge}
We calculate the quantity given by Eq.~\eqref{zeta_main} for the sequence of Hamiltonian eigenvalues 
\begin{equation}
f_j =\left\{\begin{array}{ll}
1& j\leq D-1\ ,\\
{e^{\omega j}\over 1-e^{-\omega}}& j\geq D\ .
\end{array}\right.
\end{equation}
We have that
\begin{align}
c_f&=\left(D+e^{-\omega D}\right)^{-\frac{1}{2}},\\
\epsilon_d&=\left({e^{-\omega d}\over D+e^{-\omega D}}\right)^{\frac{1}{2}}=e^{-\omega d/2}c_f,\\
s&\coloneqq \ln(2) \sigma=\ln(D+e^{-\omega D})+{\omega(1-D+De^{\omega})\over (e^\omega-1)(D e^{\omega D}+1)}-{\ln(1-e^{-\omega})\over(1+D e^{\omega D})},
\end{align}
where we assumed $d\geq D$. We then obtain \begin{align}
\zeta = \left( { 432 E^2(D+e^{-\omega D})^2 d^3 s \over N} \right)^{\!\!\frac13}\!\! + 4E \sqrt{\frac{D+e^{-\omega D}}{e^{\omega d}}},
\end{align}
which is valid for any $d\geq D$.
\end{example}

\section{Properties of the pure loss channel}
\label{appendix:b}
\subsection{Symmetry of the reduced dynamics}\label{appendix:b1}

We show that, when the dynamics from system $A$ to the environment fragments $B_1,...,B_N$ is given by
\beq \label{achannel}
\Lambda_{A\rightarrow B_1,...,B_N}(\cdot) \coloneqq U\left((\cdot)_A \otimes  \bigotimes_{j=2}^N|0\rangle \langle 0|_{B_j} \right)U^\dagger,
\eeq
with $U$ the symplectic unitary implementing a $N-$splitter,
the reduced dynamics $\Lambda_j = \Tr_{B\backslash B_j}\circ\Lambda$ have the same form for all $j$. We start by introducing the Weyl displacement operator:
 \beq 
 D(\vec{\alpha})\coloneqq\exp\left[\sum_j\left(\alpha_j a^\dagger_j - \alpha^*_j a_j\right)\right], \eeq
 where $\vec{\alpha}$ denotes a complex vector in $\mathbb{C}^N$, with $N$ the number of modes. A quantum state $\rho$ can be described in terms of the characteristic function
 \beq\chi_{\rho}(\vec{\alpha})\coloneqq\Tr\left[\rho D(\vec{\alpha})\right],\eeq
by means of which the state $\rho$ can be reconstructed as
 \beq\rho = \int \frac{d^{2N}\alpha}{\pi^N}\chi_{\rho}(\vec{\alpha}) D(-\vec{\alpha}).\eeq
 We will describe the channel in Eq.~\eqref{achannel} as the unitary operation on $\mathcal{D}(A \otimes B_2 \otimes \ldots \otimes B_N) $ given by
  \beq
 \rho_{\mathrm{in}} \rightarrow \rho_{\mathrm{out}} = U \rho_{\mathrm{in}} U^\dagger,
 \eeq
 with
 \beq
 \rho_{\mathrm{in}}=\rho_A \otimes  \bigotimes_{j=2}^N|0\rangle \langle 0|_{B_j}\ .
 \eeq
 The characteristic function for the input state is given by 
  \beq\chi_{\rho_{\mathrm{in}}}(\vec{\alpha})=\Tr\left[\rho_{\mathrm{in}} D(\vec{\alpha})\right]=\chi_{\rho_A}(\alpha_1)\chi_{|0\rangle \langle 0|}(\alpha_2)...\chi_{|0\rangle \langle 0|}(\alpha_N)=\chi_{\rho_A}(\alpha_1) \exp\left[{-\frac12\left(\|\vec{\alpha}\|^2 - |\alpha_1^2|\right)}\right]\eeq
 and for the output state we have 
  
\beq\chi_{\rho_{\mathrm{out}}}(\vec{\alpha})=\Tr\left[\rho_{\mathrm{out}} D(\vec{\alpha})\right]=\Tr\left[U\rho_{\mathrm{in}}U^\dagger D(\vec{\alpha})\right]=\Tr\left[\rho_{\mathrm{in}}U^\dagger D(\vec{\alpha})U\right].\eeq
Since $U^\dagger D(\vec{\alpha})U=D(V^\dagger\vec{\alpha})$, we obtain that
\beq\chi_{\rho_{\mathrm{out}}}(\vec{\alpha})=\Tr\left[\rho_{\mathrm{in}}D(V^\dagger\vec{\alpha})\right]=\chi_{\rho_A}\left(\frac{\alpha_1+\alpha_2+...+\alpha_N}{\sqrt{N}}\right) \exp\left[{-\frac12\left(\sumno_j|\alpha_j|^2 -\frac1N \left|\sumno_j\alpha_j\right|^2\right)}\right].\eeq
The characteristic function of the output state $\rho_{\mathrm{out}_j} = \Tr_{ B\backslash B_j} [\rho_{\mathrm{out}}]$ is obtained by setting $\alpha_{i\neq  j}=0$:
    \beq\chi_{\rho_{\mathrm{out}_j}}(\alpha_j)=\chi_{\rho_A}\left(\frac{\alpha_j}{\sqrt{N}}\right) \exp\left[-\frac12\left(\frac{N-1}{N} \right)\left|\alpha_j\right|^2\right].\eeq
It has the same form for all $j$, and the same property is therefore true for the reduced channel $\Lambda_j=\Tr_{ B\backslash B_j}\circ\Lambda$.

\subsection{Lower bound for the objectivity range of a pure loss channel}\label{appendix:b2}
The statement
\beq 
\label{nbound1}
\exists \{M_l\}_l: \forall j \in S,~ \exists\{\tau_{j,l}\}_l: ~	\|\Lambda_j - \cE_{M,\tau_{j}}  \|_{\diamond H,E} \leq {1 \over \delta} \zeta\ ,
\eeq
is equivalent to the inequality
\beq
\label{nbound2}
\inf_{M} \sup_{1 \leq j \leq N} \inf_{\tau_{j}} 	\|\Lambda_j - \cE_{M,\tau_j} \|_{\diamond H,E} \leq {1 \over \delta} \zeta\ .
\eeq
Note that we made explicit the dependence of the measure-and-prepare channels from POVM $M=\{M_l\}_l$ and set of states $\tau_j=\{\tau_{j,l}\}_l$ through the notation $\cE_{M,\tau_{j}}(X) \coloneqq  \sum_l   \Tr (M_l X) \tau_{j,l}$.
To investigate the optimality of the objectivity bound in Eq.~\eqref{nbound1} we thus need to estimate a \textit{lower} bound for the l.h.s.\ of Eq.~\eqref{nbound2}. We will perform this analysis for the channel in Eq.~\eqref{achannel}. Since the reduced dynamics $\Lambda_j$ have the same form for all $j$ we can get rid of the supremum on $j$ and look at a lower bound for the quantity 
\beq\label{mu}\mu(\Lambda) \coloneqq \inf_{M,\tau_{j}} 	\|\Lambda_j - \cE_{M,\tau_j}  \|_{\diamond H,E}\ .\eeq
By substituting the definition of diamond norm we obtain
\beq
\mu(\Lambda)=\inf_{M,\tau_j} 	\|\Lambda_j - \cE_{M,\tau_j}  \|_{\diamond H,E}
= \inf_{M,\tau_j}  \sup_{
	\substack{
		\rho :	\Tr\left[ \rho H_{A} \right] \leq E
	}
} \| \id_C \otimes (\Lambda_j - \cE_{M,\tau_j})_A  [\rho] \|_1\ ,
\eeq 
where $C$ is an arbitrary ancillary system (see Definition~\ref{diamondH} in the main text). To simplify the notation, in the following we suppress the explicit reference to the bipartition $C\!:\!A$. We can choose $\rho=\psi_r\coloneqq|\psi_r\rangle \langle \psi_r|$ with $|\psi_r\rangle\coloneqq\frac{1}{\cosh(r)} \sum_{n} \tanh(r)^n |nn\rangle$, which is a two-mode squeezed vacuum state, and (noting that $\Tr\left[ \rho H_{A} \right] = \sinh(r)^2$) find the inequality
\beq
\mu(\Lambda)\geq \inf_{M,\tau_j}  \sup_{
	\substack{
		r: \sinh(r)^2 \leq E
	}
} \| \id \otimes (\Lambda_j - \cE_{M,\tau_j})  [\psi_r] \|_1\ .
\eeq
The channel $\cE_{M,\tau_j}$ is entanglement breaking, so $ \id \otimes\cE_{M,\tau_j}[\psi_r]$ is a separable state:  $ \id\otimes\cE_{M,\tau_j}[\psi_r] = \omega \in SEP$. Since $\| X\|_1 \geq 2 \| X \|_\infty \ $ if $\Tr[X]=0$, we have that
\beq\begin{split}
	\mu(\Lambda)\geq &~ 2\inf_{\omega \in SEP} \sup_{
		\substack{
			r: \sinh(r)^2 \leq E
		}
	}   \| \id \otimes \Lambda_j [\psi_r] - \omega \|_\infty\, \\=&~2 \inf_{\omega \in SEP}\sup_{
		\substack{
			r: \sinh(r)^2 \leq E
		}
	}    \sup_{\phi} | \langle \phi |\id\otimes \Lambda_{j}  [\psi_r]|\phi \rangle - \langle \phi|\omega |\phi \rangle |\ ,
\end{split}
\eeq
where we substituted the definition of the infinity norm. We can choose, as $|\phi\rangle$, a two-mode squeezed vacuum state $|\phi_s\rangle=\frac{1}{\cosh(s)} \sum_{n} \tanh(s)^n |nn\rangle$, and get rid of the modulus to obtain
\beq
\mu(\Lambda)\geq ~2  \inf_{\omega \in SEP}  \sup_{
	\substack{
		r: \sinh(r)^2 \leq E
	}
}   \sup_{s} \left( \langle \phi_s |\id \otimes \Lambda_j  [\psi_r]|\phi_s \rangle - \langle \phi_s|\omega |\phi_s \rangle \right) .
\eeq
For a separable state $\omega$, $ |\langle \phi|\omega |\phi \rangle | \leq \lambda_{\max}$, where $\lambda$ is defined by the Schmidt decomposition: $|\phi \rangle= \sum_i \sqrt{\lambda_i}|e_i f_i \rangle$. Hence $ \langle \phi_s|\omega |\phi_s \rangle \leq \lambda_{\max} (\phi_s)= \frac{1}{\cosh(s)^2}$, and we have that
\beq
\mu(\Lambda)\geq ~ 2 \sup_{
	\substack{
		r: \sinh(r)^2 \leq E
	}
}\sup_{s}   \left( \langle \phi_s |\id \otimes \Lambda_j [\psi_r]|\phi_s \rangle -\frac{1}{\cosh(s)^2} \right).
\eeq
A calculation of the quantity $ \langle \phi_s |\id \otimes \Lambda_j [\psi_r]|\phi_s \rangle$ can be found in~\cite{extendibility}. By exploiting that result we find
\beq
\mu(\Lambda)\geq 2  \sup_{s} \left[ \left( \sup_{
	\substack{
		r: \sinh(r)^2 \leq E
	}
} \frac{N}{(\sqrt{N}\cosh(r)\cosh(s)- \sinh(r)\sinh(s))^2}\right) -\frac{1}{\cosh(s)^2} \right] .
\eeq
For a given $s$, the supremum of the function
\beq
\frac{N}{(\sqrt{N}\cosh(r)\cosh(s)- \sinh(r)\sinh(s))^2} 
\eeq
is reached for $r=\bar{r}$ such that $\bar{E} \coloneqq \sinh(\bar{r})^2 = \frac{\tanh(s)^2}{N - \tanh(s)^2}$. Since $\bar{E} \le\frac{1}{N-1} \le \frac{2}{N}$ for $N\geq2$ (and noting that $N\geq2$ by definition of the channel $\Lambda$), we can choose $E\ge\frac{2}{N}$ in order to have $\bar{E} \leq E$ satisfied for all possible values of $N$. This is equivalent to evaluate an unconstrained supremum, for which we can use the calculation performed in~\cite{extendibility} to obtain
\beq\begin{split}
\mu(\Lambda)&\geq 2  \sup_{s} \left[ \frac{N}{N\cosh(s)^2-\sinh(s)^2} -\frac{1}{\cosh(s)^2} \right] 
= 2  \sup_{s} \frac{\tanh(s)^2}{N\cosh(s)^2-\sinh(s)^2} \\ &\geq \dfrac{1}{2N-1}  \ge \frac{1}{2N} .
\end{split}
\eeq
Our analysis therefore led to the following result: when the dynamics from $A$ to $B_1,...,B_N$ is given by Eq.~\eqref{achannel} and the maximum energy of system $A$ satisfies $E\ge \frac{2}{N} $, for all $j$ and for all POVM $\{M_l\}_l$ and sets $\{\tau_{j,l}\}$ entering the definition of $\cE_j$ it holds that
\beq 
\|\Lambda_j- \cE_j  \|_{\diamond H,E} \ge \frac{1}{2N} .
\eeq
\section{Proof of Corollary~\ref{cor}} \label{appendix:c}

\renewcommand*{\thetheorem}{C\arabic{theorem}}
\renewcommand*{\thetmain}{M\arabic{tmain}}
\renewcommand*{\theexample}{C\arabic{example}}
\renewcommand*{\thelemma}{C\arabic{lemma}}
\renewcommand*{\thecorollary}{C\arabic{corollary}}
\renewcommand*{\theremark}{C\arabic{remark}}
\renewcommand*{\theobservation}{C\arabic{observation}}
\renewcommand*{\theproposition}{C\arabic{proposition}}
\renewcommand*{\thedefinition}{C\arabic{definition}}

To prove Corollary~\ref{cor}, it suffices to adapt to our infinite-dimensional setting the argument in the proof of~\cite[Corollary~4]{Brandao2015QD}. The success of this programme depends crucially on a fundamental result by Winter~\cite[Lemma~17]{tightuniform}, reported below as Lemma~\ref{meta}, which expresses a continuity bound for the conditional entropy of infinite-dimensional systems subjected to energy constraints. 

\begin{lemma}[{\cite[Lemma~17]{tightuniform}}] \label{meta}
 	For a Hamiltonian $H$ on $A$ satisfying the Gibbs hypothesis and any two states $\tau$ and $\sigma$ on the bipartite system $A \otimes B$ with $\Tr(\tau H), \Tr(\sigma H) \le E$, $\frac{1}{2}\|\tau - \sigma\|_1 \le \epsilon < \epsilon' \le 1$ and $\Delta = {\epsilon' - \epsilon \over 1 + \epsilon '}$,
 	\beq 
 	|S(A|B)_\tau - S(A|B)_\sigma | \le  (2 \epsilon' + 4\Delta) S(\gamma(E/\Delta)) + (1 + \epsilon')\,h\left({\epsilon' \over {1 + \epsilon'} }\right) + 2 h(\Delta)\ .
 	\eeq
\end{lemma}

\begin{proof}[Proof of Corollary~\ref{cor}]
Let $0<\delta<1$ be a fixed number. Theorem~\ref{main_th_f} allows us to construct a POVM $\{M_l\}_l$, a set $S\subseteq \{1,\ldots, N\}$ of cardinality at least $|S|\geq (1-\delta)N$, and ensembles of states $\{\tau_{j,l}\}_l$ such that the corresponding measure-and-prepare channels $\cE_j$ defined in Eq.~\eqref{m&p channel} satisfy Eq.~\eqref{almost m&p} and~\eqref{zeta_main} for all $j\in S$. Now, consider the states $\tau = (\id \otimes \Lambda_j)(\rho)$ and $\sigma= (\id \otimes \cE_j)(\rho)$. By definition of $f$-diamond norm it follows that
\begin{align*}
 \frac{1}{2}\|\tau - \sigma\|_1 &= \frac12 \|(\id \otimes \Lambda_j)(\rho) - (\id \otimes \cE_j)(\rho)\|_1 \\
 &\le \frac12 \|\Lambda_j - \cE_j\|_{\Diamond H_B,E_B} \\
 &\le \epsilon  < \epsilon' \le 1\, ,
\end{align*}
where the inequalities in the last line follow from Theorem~\ref{main_th_f}, and we set $\epsilon' \coloneqq 2 \epsilon \coloneqq {\zeta \over \delta}$, with $\zeta$ given in Eq.~\eqref{zeta_main}.

Applying Lemma~\ref{meta} to states $\tau$ and $\sigma$, we deduce that
\beq \label{S}
|S(A|B)_{\id \otimes \Lambda_j(\rho)} - S(A|B)_{\id \otimes \cE_j(\rho)} | \le  (2 \epsilon' + 4\Delta) S(\gamma(E_A/\Delta)) + (1 + \epsilon')\,h\left({\epsilon' \over {1 + \epsilon'} }\right) + 2 h(\Delta)
\eeq
where $\Delta \coloneqq \frac{1}{2}{\epsilon' \over 1 + \epsilon '}$. Since the reduced entropies on the $A$ subsystems are the same for $\tau$ and $\sigma$, this translates to
 \begin{align*}
 &\left|I(A:B)_{(\id \otimes \Lambda_j)(\rho)} - I(A:B)_{(\id \otimes \cE_j)(\rho)} \right| \\ &\quad \le  (2 \epsilon' + 4\Delta) S(\gamma(E_A/\Delta)) + (1 + \epsilon')\, h\left({\epsilon' \over {1 + \epsilon'} }\right) + 2 h(\Delta) ~,
 \end{align*}
 and therefore
 \begin{align*}
 &I(A:B)_{\id \otimes \Lambda_j(\rho)} \\ & \quad \le I(A:B)_{\id \otimes \cE_j(\rho)}  +  (2 \epsilon' + 4\Delta) S(\gamma(E_A/\Delta)) + (1 + \epsilon')\, h\left({\epsilon' \over {1 + \epsilon'} }\right) + 2 h(\Delta) \\
 & \quad\le \max_{\Gamma \in QC} I(A:B)_{(\id \otimes \Gamma)(\rho)}  +  (2 \epsilon' + 4\Delta) S(\gamma(E_A/\Delta)) + (1 + \epsilon')\, h\left({\epsilon' \over {1 + \epsilon'} }\right) + 2 h(\Delta)
\, ,
 \end{align*}
where the last inequality follows because any measure-and-prepare channel can be obtained by post-processing from a quantum-to-classical channel, and the mutual information obeys the data processing inequality.

We now move on the proof of Eq.~\eqref{rs}. The fact that the right hand side is no larger than the left hand side is well known; to prove it, it suffices to choose as $\Lambda$ the quantum-to-classical map that attains the accessible information $I(A:B_a)\coloneqq \max_{\Gamma \in QC} I(A:B)_{(\id \otimes \Gamma)(\rho)}$, makes $N$ copies of the classical result, and stores it in $N$ registers $B_1\ldots B_N$. 

As it turns out, we only have to prove that the l.h.s.\ of Eq.~\eqref{rs} is no larger than the the r.h.s.\ .
By using Eq.~\eqref{dis} we can write
 \begin{align}
 	&\bbE_j 	I(A:B_j) \\ & \quad\le \frac{1}{N}\mbox{$\left[(1-\delta)N   \left(I(A:B_a)  +  (2 \epsilon' + 4\Delta) S(\gamma(E_A/\Delta)) + (1 + \epsilon')\,h\left({\epsilon' \over {1 + \epsilon'} }\right) + 2 h(\Delta)\right) + \delta N 2S(A) \right]$} \nonumber\\ &\quad=\mbox{$(1-\delta) \left(I(A:B_a) + (2 \epsilon' + 4\Delta) S(\gamma(E_A/\Delta)) + (1 + \epsilon')\, h\left({\epsilon' \over {1 + \epsilon'} }\right) + 2 h(\Delta)\right) + \delta  2S(A)$}\ ,\nonumber 
 \end{align}
where we used the notation $I(A:B_j) \coloneqq I(A:B)_{\id \otimes \Lambda_j(\rho)}$. We can choose $\delta = \sqrt{\zeta}$, then
 \beq
 \epsilon' = 2 \epsilon={\zeta \over \delta} \xrightarrow[N \to \infty]{} 0~,
 \eeq
 and therefore
 \beq
 \Delta = \frac{1}{2}{\epsilon' \over 1 + \epsilon '}\xrightarrow[N \to \infty]{} 0~.
 \eeq
 Moreover, since $S(\gamma(E_A))=o(E_A)$ ~\cite{tightuniform}, we have that $\Delta\: S(\gamma(E_A/\Delta)) \xrightarrow[\Delta \to 0]{} 0$, as well as $\epsilon' S(\gamma(E_A/\Delta)) \xrightarrow[\epsilon',\Delta \to 0]{} 0$ (since $\epsilon' = O(\Delta)$). As a consequence, for our choice of $\delta$,
 \begin{align} 
& \bbE_j 	I(A:B_j) \nonumber \\ & \quad  \le\mbox{$(1-\delta) \left(I(A:B_a)  +  (2 \epsilon' + 4\Delta) S(\gamma(E_A/\Delta)) + (1 + \epsilon')h({\epsilon' \over {1 + \epsilon'} }) + 2 h(\Delta)\right) + \delta  2S(A)$} \nonumber \\ & \quad \xrightarrow[N \to \infty]{}I(A:B_a)\,,
 \end{align}
 independently of the choice of $\Lambda=\Lambda_{B \rightarrow B_1 B_2 ... B_N}$. By considering the maximum of $\bbE_j I(A:B_j)$ over $\Lambda_{B \rightarrow B_1 B_2 ... B_N}$ and then the limit $N \to \infty$ we therefore obtain that
 \beq
 \lim_{N \to \infty} \max_{\Lambda_{B \rightarrow B_1 B_2 ... B_N}} \bbE_j 	I(A:B_j) \le  I(A:B_a).
 \eeq
 \end{proof}


\end{document}